%% file: main.tex
\documentclass[12pt]{article}

\usepackage[margin=1in]{geometry}
\usepackage{graphicx}
\usepackage{amsmath}
\usepackage{amsfonts}
\usepackage{amsthm}
\usepackage{hyperref}

\usepackage{algorithm}
\usepackage{algorithmic}
\usepackage{bm}

\newtheorem{example}{Example}[section]
\newtheorem{theorem}{Theorem}[section]
\newtheorem{lemma}[theorem]{Lemma}

\newtheorem{corollary}[theorem]{Corollary}

\newcommand{\RR}{\mathbb{R}}      

\newcommand{\A}{\mathcal{A}}
\newcommand{\XX}{\mathcal{X}}
\newcommand{\PP}{\mathcal{P}}

\DeclareMathOperator*{\argmax}{arg\,max}

\usepackage[usenames,dvipsnames,svgnames,table]{xcolor}
\newcommand\rcb[1]{}
\newcommand\jm[1]{}
\newcommand\os[1]{}

\def\sepobjective{0.8em}
\newenvironment{program}{\begin{displaymath}\begin{array}{cr@{\hspace{3pt}}l@{\hspace{2em}}{l}}}{\end{array}\end{displaymath}}
\def\maximize{\vspace{\sepobjective} \mathrm{maximize}}
\def\minimize{\vspace{\sepobjective} \mathrm{minimize}}
\def\subjectto{\mathrm{subject\ to}}

\begin{document}


\title{The Ad Types Problem}

\author{Riccardo Colini-Baldeschi\thanks{Facebook, Core Data Science, \texttt{rickuz@fb.com}} \and Juli\'{a}n Mestre\thanks{University of Sydney, School of Computer Science, \texttt{julian.mestre@sydney.edu.au}} \and Okke Schrijvers\thanks{Facebook, Core Data Science, \texttt{okke@fb.com}} \and Christopher A. Wilkens\thanks{Tremor Technologies, \texttt{c.a.wilkens@gmail.com}. The work was done while the author was a full-time employee at Facebook Core Data Science.}}

\maketitle

\begin{abstract}
The Ad Types Problem (without gap rules) is a special case of the assignment problem in which there are $k$ types of nodes on one side (the ads), and an ordered set of nodes on the other side (the slots). The edge weight of an ad $i$ of type $\theta$ to slot $j$ is $v_i\cdot \alpha^{\theta}_j$ where $v_i$ is an advertiser-specific value and each ad type $\theta$ has a discount curve $\alpha^{(\theta)}_{1} \ge \alpha^{(\theta)}_{2} \ge ... \ge 0$ over the slots that is common for ads of type $\theta$. We present two contributions for this problem: 1) we give an algorithm that finds the maximum weight matching that runs in $O(n^2(k + \log n))$ time for $n$ slots and $n$ ads of each type---cf. $O(kn^3)$ when using the Hungarian algorithm---, and 2) we show to do VCG pricing in asymptotically the same time, namely $O(n^2(k + \log n))$, and apply reserve prices in $O(n^3(k + \log n))$.

The Ad Types Problem (with gap rules) includes a matrix $G$ such that after we show an ad of type $\theta_i$, the next $G_{ij}$ slots cannot show an ad of type $\theta_j$. We show that the problem is hard to approximate within $k^{1- \epsilon}$ for any $\epsilon > 0$ (even without discount curves) by reduction from Maximum Independent Set. On the positive side, we show a Dynamic Program formulation that solves the problem (including discount curves) optimally and runs in $O(k\cdot n^{2k + 1})$ time.
\end{abstract}


\input{intro}
\input{prelim}
\input{discount-only}
\input{gap-rules}

\bibliographystyle{plain}
\bibliography{main}


\end{document}

%% file: intro.tex
\section{Introduction}
\label{sec:intro}

In online advertising, there is a hugely influential literature on sponsored search \cite{V07, EOS07, LPSV07}. In this setting, a small number of ads (typically 2 or 3) are shown next to organic search results using a \emph{position auction} \cite{V07, EOS07}. Since users tend to look at the results from top to bottom, an ad will have a higher click-through-rate (CTR) in a higher slot, and are therefore higher slots rae considered more valuable. The common way to capture this slot effect is with a separable model, where ad slots have an associated discount $1\geq \alpha_1 \geq .. \geq 0$ that represents the advertiser-agnostic CTR of the slot. The auction is run taking the expected number of clicks (as represented by $\alpha_j$) as the ``quantity of good'' that is sold. Work in this area has proliferated and the position auction model lies at the basis of many applications of online advertising, even outside of sponsored search. 

While the position auction model has been very influential in online advertising, there are  important differences between the multi-slot ad allocation problem on a search results page and the multi-slot problem in content feeds like Facebook's News Feed, reddit, and Apple News. In content feeds, ads are interspersed with organic content, so not all ads are visible from the start, and when the first ad is shown it is not clear how many ads will have to be allocated.\footnote{While this may not be a problem when deciding which ad to show the user, pricing generally needs information about how many items are sold, and what the other competition is in the auction.} Moreover, there are many different types of organic content like text, photos, and videos mixed in the same stream. Similarly, there are many different types of ads that have different associated objectives: some ads simply want the user to see them, others want the user to click or view a video. These objectives must be considered when allocating slots to ads, since the probability that a user will see an impression-ad is quite different from the probability that she'll watch a video ad or click on a link-click ad. There may be diversity rules as well that after a video ad you cannot show another video ad, but it's okay to show a banner ad. These considerations have implications on which ads are eligible to be shown in which slots.

The standard position auction model fails to capture these important differences. For example, the probability of a user watching a video ad typically decays differently than a link-click ad. If one were to use the same discount curve for ads of different types, this could lead to a suboptimal allocation:

\begin{example}
Suppose we have a setting with 2 ad types, link-click ads and video ads, two ad slots, and we have discount curve $\alpha_1 = \tfrac12, \alpha_2=\tfrac14$. These discounts are accurate for link-click ads (i.e. $\alpha^{(\text{link})}_1 = \tfrac12, \alpha^{(\text{link})}_2=\tfrac14$), but for video ads, the user is more likely to watch the video in the second slot than they are to click a link in that slot: $\alpha^{(\text{video})}_1 = \frac12$, $\alpha^{(\text{video})}_2 = \frac13$. 

Consider a video ad with bid $\$12$ and a link-click ad with bid $\$10$. The optimal allocation assuming that discount curve $\alpha$ is accurate for both ads would assign the video ad to slot 1 and the link-click ad to slot 2 for total value $\tfrac{1}{2}\cdot\$12+\tfrac{1}{4}\cdot\$10 = \$8.50$. However, switching the ads yields total value $\tfrac{1}{2}\cdot\$10+\tfrac{1}{3}\cdot\$12 = \$9 > \$8.50$.\footnote{Those familiar with auction theory will note that this also implies that VCG prices w.r.t. $\alpha$ would not be incentive compatible. For readers less familiar with auction theory, we provide a primer in Section~\ref{ss:auction-theory}.}

\end{example}

In this paper we propose a new theoretical model for online advertising that addresses these issues. It captures the position auction as a special case, but can handle discount curves for multiple types and intersperse advertising with organic content in a dynamic manner.\footnote{While our motivation for studying this problem comes from online advertising in content streams, it captures many other interesting settings that are unrelated to online advertising. For example, the setting without gap rules can model a worker with different time slots and jobs of different types that need to be done; jobs are most valuable when completed early and delays for jobs of the same type are discounted similarly. Adding gap rules can model the cost of moving between locations (in the physical world) or context-switching (in the digital world).
} An \emph{Ad Types Problem} instance has $k$ ad types, that each have their own discount curve over $n$ slots, i.e. for each ad type $\theta$ we have a discount curve $\alpha^{(\theta)}_1 \geq \alpha^{(\theta)}_2 \geq \ldots \ge \alpha^{(\theta)}_n \geq 0$ that represents the slot-specific action-rate for ads of type $\theta$. All ad types agree on the order of the slots. Gap rules are modeled by a $k\times k$ matrix $G$, which indicates for each pair of ad types $(\theta_i, \theta_j)$, that after showing an ad of type $\theta_i$, the next $G_{ij}$ stories cannot be of type $\theta_j$.

We first focus on the special case where $G = \bf 0$, i.e. different ad types have different discount curves but there are no constraints on the gaps between ads. In this setting, the Ad Types Problem is a special case of the maximum-weight bipartite matching problem (also known as the assignment problem), so we could find an optimal allocation using the Hungarian algorithm in $O(kn^3)$ time \cite{RT12} (where $k$ is the number of ads, $n$ the number of slots, and we have $n$ ads per type). Our first result is an algorithm that finds the optimal allocation in $O(n^2(k + \log n))$ time, saving a linear factor. In auctions, it's not only important to find the optimal allocation of ads to slots, one also needs to compute appropriate prices. We show that we can compute VCG prices in asymptotically the same time as finding the optimal allocation, i.e. $O(n^2(k + \log n))$. A common way to control revenue is to set an appropriate minimum bid (also knows as a \emph{reserve price} \cite{M81}); we show that we can compute incentive-compatible prices with advertiser-specific reserve prices for all ads in $O(n^3(k + \log n))$ time.

Next we consider the more general Ad Types problem with both discount curves and gap rules (where $G \neq \bf 0$). We show that the problem is hard to approximate within $k^{1- \epsilon}$ for any $\epsilon > 0$ (even without discount curves) by reduction from Maximum Independent Set. On the positive side, we show a Dynamic Program formulation that solves the problem (including discount curves) optimally and runs in $O(k\cdot n^{2k + 1})$ time.

\subsection{Related Work}
\paragraph{Assignment Problem.}

The maximum-weight bipartite matching problem, also known as the assignment problem, is a classical problem in operations research. Let $(A, B, E)$ be a complete bipartite graph with edges weights $v:E \rightarrow \RR^+$, and $V=A \cup B$ the set of nodes; the goal is to find a matching $M$ of maximal total weight $\sum_{e\in M} v(e)$. Kuhn \cite{K55} proposed an algorithm for this problem---which he called the Hungarian algorithm---based on ideas by K\H{o}nig and Egerv\'{a}ry, though he only proved that the algorithm would terminate, not what the time complexity is. Munkres \cite{M57} showed that the time complexity of the Hungarian algorithm is $O(|V|^4)$. Edmonds and Karp \cite{EK72} gave a $O(|V|^3)$ time algorithm for balanced graphs, and Ramshaw and Tarjan \cite{RT12} more recently gave an algorithm for unbalanced graphs (wlog assume $|A|<|B|$) that runs in $O(|E||A| + |A|^2\log |A|)$.
Since the seminal work on the assignment problem, there has been active research into relevant special cases. In particular there is a line of work on \emph{convex bipartite graphs}, where the right side of the graph is ordered, and nodes on the left can only be connected to a single contiguous block of nodes on the right. For the unweighted case, a line of work starting with Glover \cite{G67, LP81, GT85} shows that the problem can be solved in time linear in the number of nodes $O(|V|)$. General weights are not considered, though early work on \emph{vertex-weighted} bipartite graphs (where each node $i$ has an associated weight $w_i$ and the weight of an edge from $i$ to $j$ is $w_{ij} = w_i + w_j$) yield a $O(|E| + |B|\log|A|)$ time algorithm \cite{K08}. More recently, Plaxton \cite{P08,P13} showed that Two-Directional Orthogonal Ray Graphs (a generalization of convex graphs) admit a $O(|V| \log |V|)$ time algorithm.

Sharathkumar and Agarwal \cite{SA12} consider a more general set of edge weights, where nodes are embedding in $d$-dimensional space, and the weights of the complete bipartite graph are all either the $L_1$ or $L_\infty$ metric. They present an algorithm to solve the maximum weight bipartite in $O(|V|^{3/2}\log^{d+O(1)}(|V|)\log \Delta)$, where $\Delta$ is the diameter of the space that contains the points.

None of the results on specializations cover The Ad Types Problem setting (even without gap rules).

\paragraph{Ad Auctions.}
The simple separable model for position auctions appears in Varian \cite{V07} and Edelman et al. \cite{EOS07}. One body of related work relaxes the assumption that action rates are separable. One common theme is to model externalities between ads (also related to our gap rules) \cite{KM08, GM08, GK08, AE11, AFMP08, GIM09, GS10}. Of note, \cite{KM08, GM08, FG16, AFMP08} study algorithms for computing allocations in models where the user's attention cascades and prove hardness results. A different generalization is to allow arbitrary action rates that are still independent between ads \cite{AGV07, CW14, CSW18}, which corresponds to the Ad Type Problem (without gap rules) where each ad has a unique type.

Another generalization of the basic position auction allows ads to be placed in complex ways. A few papers study mechanisms that permit presentation constraints and/or ads with variable presentation \cite{H16, CKSW17, HIKLN18, MN09, DSYZ10}. Mahdian et al. study auctions for ads displayed on maps along with organic results \cite{MSV15} (since places of interest are connected to a physical location, this imposes constraints on where ads can be placed).

Finally, the connections between ad auctions and max-weight matching (and the Hungarian algorithm) have been studied before as well \cite{DHW13, CSW18, KMU16, EK10}.

\subsection{Contributions}

This paper presents three main contributions:
\begin{itemize}
    \item {\bf Optimal Allocation.} Firstly, we give an algorithm to optimally solve the Ad Types problem \emph{without} gap rules. This setting is a special case of the assignment problem with applications beyond ad auctions. Our algorithm is a specialization of the Hungarian algorithm to find the maximum-weight matching in the bipartite graph that uses the structure of the Ad Types Problem to run in $O(n^2(k + \log n))$ time (compared to $O(kn^3)$ for running the Hungarian algorithm on the instance; Theorem~\ref{thm:runtime}). While we are motivated by ad auctions, this setting also models other common settings like a worker who needs to do several tasks of different types.
    \item {\bf Pricing.} Secondly, we show that we can do incentive-compatible pricing in this setting with minimal overhead. First we show how to compute VCG prices efficiently. Naively, VCG prices are computed by resolving the original problem $n$ times (removing 1 buyer in each instance). We show that we can extract VCG prices from the dual variables in our algorithm when they are properly updated. This yields VCG prices in $O(n^2(k + \log n))$ time (cf. $O(n^3(k + \log n)$ when implemented naively; Theorem~\ref{thm:vcg}). Second, we show that we can apply reserve prices (and in fact in all single-parameter environments) without a changepoint algorithm \cite{M81, R16}. For our case, this yields a $O(n^3(k + \log n))$ time algorithm (Corollary~\ref{cor:myerson}). 
    \item {\bf Gap Rules.} Finally, we consider the more general Ad Types problem with both discount curves and gap rules (where $G \neq \bf 0$).  We show that the problem is hard to approximate within $k^{1- \epsilon}$ for any $\epsilon > 0$ (even without discount curves) by reduction from Maximum Independent Set (Theorem~\ref{thm:hard}). On the positive side, we give a Dynamic Program formulation that solves the problem (including discount curves) optimally and runs in $O(k\cdot n^{2k+1})$ time (Theorem~\ref{thm:dp}).
\end{itemize}

%% file: prelim.tex

\section{Preliminaries}
\label{preliminaries}

In this section we give a primer on auction theory, a formal definition of the ad types problem, and a review of the Hungarian Algorithm and its application to Ad Types problem.

\subsection{Auction Theory}\label{ss:auction-theory}
There are $n$ (strategic) agents, and some finite set of outcomes $\Omega$. Agents have some value function over outcomes: $v_i : \Omega \rightarrow \RR$. An auction $\A = (\XX, \PP)$ consists of an allocation function $\XX : ( \Omega \rightarrow \RR)^n \rightarrow \Omega$ that takes the bids (meaning a \emph{report} of the valuation) of $n$ agents and selects the outcome, and a payment function $\PP : ( \Omega \rightarrow \RR)^n \rightarrow \RR$.

A common goal\footnote{Other common goals include revenue and a notion of high-value allocation with fairness concerns.} for the auctioneer is to select the outcome $\omega^*$ with the maximal total value, which is known as the \emph{social welfare}: $\omega^* = \argmax_{\omega \in \Omega} \sum_{i=1}^n v_i(\omega)$. Note however, that the auctioneer doesn't have direct access to the valuation function $v_i$; she elicits them as bids $b_i$. Each agent cares about their own utility, $u_i(b_i, b_{-i}) = v_i(\XX(b_i, b_{-i})) - \PP(b_i, b_{-i})$, where $b_{-i}$ are the bids of all agents except $i$. Thus, given bids by the other agents $b_i$ she will report $b^*_i = \argmax_{b_i} u_i(b_i, b_{-i})$. The auctioneer can use the pricing function to incentivize agents to report their true valuations. An auction is called \emph{dominant strategy incentive compatible (DSIC)}, or simply incentive compatible, if for any bid profile $b_{-i}$ of agents other than $i$, agent $i$ always maximizes their utility by bidding truthfully: $v_i \in \argmax_{b_i} u_i(b_i, b_{i-1}))$. Celebrated work by Vickrey, Clarke and Groves \cite{V61, C71, G73} shows that it is always possible to do this by charging agents their externality.

A particularly interesting class of valuation functions are \emph{single-parameter} in the following sense: there is some function $x_i : \Omega \rightarrow \RR^+$ that maps outcomes to a ``quantity of goods'' that agent $i$ receives. The valuation function is then expressed as $v_i(\omega) = v_i\cdot x_i(\omega)$, where we overload $v_i$ to be a single ``value-per-unit-good'' parameter. Single-parameter settings are quite common, e.g. single-item auction are single parameter: an agent has a value $v_i$ for receiving the good, the outcome space consists of the item going to each of the agents that participate, and $x_i(\omega) = 1$ iff agent $i$ receives the good in outcome $\omega$ and it is 0 otherwise. For single-parameter environments the allocation function $\XX$ and pricing function $\PP$ take as input a single number for each advertiser: their per-unit bid $b_i\in\RR$.

Besides social welfare, a common goal for an auctioneer is to maximize revenue. Myerson \cite{M81} shows how to maximize expected revenue in single-parameter environments when a buyer has knowledge about the distributions $F_1, F_2, ..., F_n$ from which each buyer draws their value. For a buyer $i$ whose distribution $F_i$ is \emph{regular}---meaning $\phi_i(v) := v - \frac{1-F_i(v)}{f_i(v)}$ with CDF $F_i$ and corresponding PDF $f_i$ is monotonically non-decreasing---the optimal expected revenue from that buyer is obtained by setting reserve price $r_i$ (also known as a minimum bid) equal to $\phi^{-1}(0)$.\footnote{Myerson also extends the results to non-regular distributions, though that is outside the scope of the present work. Setting the optimal reserve price for a non-regular distribution is still guaranteed to be within a factor 2 of optimal, see e.g. \cite{RS16}.} To apply a reserve price $r_i$, let $x_i(b_i;b_{-i})$ be the quantity of goods that a single-parameter buyer receives. If $b_i < r_i$ the buyer is rejected from the auction and thus receives nothing; otherwise $x_i$ is identical to the setting without reserve price $r_i$. The incentive-compatible price corresponding to $x_i$ is given by Myerson's payment identity: $p_i = \int_{b = 0}^{b_i} b\cdot x'_i(b;b_{-i}) \text{ d}b$. For discrete $x_i$, this becomes a summation over the changepoints where the allocation changes.

\subsection{The Ad Types Problem}

\emph{The Ad Types Problem} involves computing an allocation of a set of $N$ ads to $n \leq N$ slots. Ads come in one of $k$ different ad types $\theta_l$, for $l\in\{1, ..., k\}$. We let the ads of type $\theta_l$ be $a^{(\theta_l)}_i$ for $i \in {1, \ldots, n_l}$. There are three main components to the definition of the problem:

\begin{itemize}
	\item {\bf Valuations.} Ad $i$ of type $\theta$ has a value-per-conversion (a.k.a. value-per-action) $v^{(\theta)}_i$. Ads of different types have different conversion events, e.g. for a display ad the conversion event is a view, for a link ad the conversion event is a link click, and for a video ad the conversion event is the user watching video ad. For each ad type $\theta$, we index the ads in non-increasing order of valuation, i.e. $v^{(\theta)}_1 \ge v^{(\theta)}_2 \ge \ldots \ge v^{(\theta)}_{n_l} \ge 0$.

	\item {\bf Discount curves.} We assume a separable model for discount curves where we can write $\Pr[\text{conversion on ad }i\text{ (of type }\theta\text{) in slot }j] = \alpha^\theta_j \cdot \beta_i$ where $\alpha^\theta_j$ is the slot effect for a particular ad type $\theta$ (e.g., the probability that a user will watch a video ad if its shown in the $j$th slot) and $\beta_i$ is the advertiser quality (this separable model is also standard in position auctions \cite{V07,EOS07}). In the remainder of the paper we assume wlog that the advertiser effect has already been included in the advertiser's value, i.e., if the value-per-conversion of the advertiser is $v_i'$, then $v_i = \beta_i\cdot v'_i$. We further abuse notation to let $v_{ij} = \alpha^\theta_j\cdot v_i$ for ad $i$ of type $\theta$ in slot $j$.

	Discounts are monotonically non-increasing, and all ad types agree on the order of slots, i.e. for each ad type $\theta$, we have $\alpha^{(\theta)}_1 \geq \alpha^{(\theta)}_2 \geq \ldots \ge \alpha^{(\theta)}_n \geq 0$.

	\item {\bf Gap rules.} When ads are interspersed with organic content, there must be some way to control how many ads are shown. In the simplest case, where there's only one type of ad, this can be implemented by a gap rule $g$, which states that two ads must be at least $g$ slots apart from each other. When there are multiple ad types, there is a $k\times k$ matrix $G$, which indicates for each pair of ad types $(\theta_i, \theta_j)$, that after showing an ad of type $\theta_i$, the next $G_{ij}$ stories cannot be of type $\theta_j$.
\end{itemize}

The Ad Types Problem is to find a social welfare maximizing allocation that obeys the gap rules.

\subsection{The Hungarian Algorithm}
\os{In the related work I use $A$ and $B$, let's make this consistent.}
\os{Do pass over this section later.}

\jm{Sorry, I changed this because $A$ and $B$ meant something different in the next section ($A$ is the augmenting path, and $B$ is the alternating tree}
The Hungarian Algorithm \cite{K55, M57} is a classical algorithm for computing a maximum weight matching in a bipartite graph. Starting from a trivial primal solution (empty matching) and a trivial dual solution, the algorithm iteratively increases the cardinality of the matching while improving the value of the dual solution until the value of the primal solution equals that of the dual.

Let $(U, V, E)$ be a complete bipartite graph with edges weights $v:E \rightarrow \RR^+$. The primal/dual pair of linear programs capturing the problem are as follows.

\begin{program}
	\maximize & \sum_{(i, j) \in E} & v_{i j} x_{i j} \\
	\subjectto
	& \sum_{j} x_{ij} & \leq 1 &  \forall i \in U \\
	& \sum_{i} x_{ij} & \leq 1 & \forall j \in V \\
	& x_{ij} & \geq 0 & \forall (i, j) \in E
\end{program}

\begin{program}
	\minimize & \sum_{i \in U} u_i  & + \sum_{j \in V} p_j \\
	\subjectto
	& u_i + p_j & \geq v_{ij} & \forall (i,j) \in E \\
	& u_{i} & \geq 0 & \forall i \in U \\
	& p_{j} & \geq 0 & \forall j \in V
\end{program}

The algorithm starts from an empty primal solution $M = \emptyset$, and a trivial feasible dual solution $u_i = 0$ for all $i \in U$ and $p_j = \max_{(i,j) \in E} v_{ij}$ for all $j \in V$. In each iteration, the algorithm identifies the set of tight edges $T = \{ (i, j) \in E : u_i + p_j = v_{ij}\}$ and builds an alternating BFS tree $B$ (also known as Hungarian tree) in $(U, V, T)$ out of the free vertices in $V$. If the alternating tree contains an augmenting path $A$, we augment $M$ with $A$ thus increasing its cardinality; if no such path is available, we can update the dual solution by reducing the dual value of $V \cap B$ and increasing the dual value of $U \cap B$ by the same amount until a new edge become tight. This update maintains feasibility while reducing the value of the dual solution and makes at least one new edge tight, which in turn allows us to grow the alternating tree further.

Throughout the execution of the algorithm we maintain the invariants that the dual solution is feasible and that the edges in the matching $M$ are tight. As a result, at the end of the algorithm we have a matching whose weight equals the value of the dual feasible solution, which acts as a certificate of its optimality.

Using the right data structures, it is possible to implement the algorithm so that the amount of work done between each update to $M$ is $O(|E| + |U| \log |U|)$. Therefore, if we let $M^*$ be a maximum weight matching, then the Hungarian algorithm can be implemented to run in $O(|M^*| (|E| + |U| \log |U|) )$ time~\cite{FT87}.

%% file: discount-only.tex

\section{Ad Types Problem without Gap Rules}
\label{discount-only}

In this section we consider the ad types problem with discount curves but no gap rules. In this model we have $k$ ad types, and each ad type has its own monotonically decreasing discount curve $\alpha_j^{(\theta_l)}$ for $l \in 1, 2, \ldots, k$. Without gap rules, the problem becomes a simple maximum weight bipartite assignment on a complete graph with $N$ vertices (ads) on one side of the bipartition and $n$ vertices (slots) on the other side of the bipartition, with $n < N$. Therefore, the Hungarian algorithm can solve this problem in $O(N n^2)$ time. We will assume throughout there are exactly $n$ ads of each type\footnote{If an ad type has fewer than $n$ ads, we can append ads with value $0$, if there are more than $n$ ads of a type, with loss of generality we can restrict attention to the $n$ highest-value ads.}, hence the Hungarian algorithm runs in $O(kn^3)$ time.

In this section we start by giving an algorithm that finds the maximum-weight bipartite matching in $O(n^2(k + \log n))$ time (Section~\ref{ss:opt-problem}). We show that in some sense the dependency on $k$ can be expected: when $k = n$, the problem is no easier than solving the assignment problem (i.e. a common order of the slots does not improve the running time; see Section~\ref{ss:kn}). We then turn our attention to pricing and show that computing VCG prices can be done in $O(n^2(k + \log n))$ (Section~\ref{ss:vcg} and that we can apply reserve prices in $O(n^3(k + \log n))$ (Section~\ref{ss:reserve-prices}).

\subsection{Finding the Optimal Allocation}
\label{ss:opt-problem}

We present an adaptation of the Hungarian algorithm \cite{K55, M57} that exploits the special structure of the Ad Types problem. In the following we use the language of markets to describe the Hungarian algorithm: the dual variable of a slot $j$ corresponds to a price $p_j$, while a dual variable of an advertiser $i$ corresponds to the utility $u_i$ of the advertiser if they get an item out their demand set (given the prices) \cite{EK10}. Moreover, the instance is a complete bipartite graph with ads on one side and slots on the other side where the weight of the edge $(i,j)$ is $v_{ij}$. The maximum-weight matching in the bipartite graph corresponds to the social-welfare maximizing allocation of ads to slots.
For ease of exposition, we assume that values and discounts are monotonically \emph{strictly} decreasing, this restriction can be lifted by consistent tie-breaking.

\begin{algorithm}
	\begin{algorithmic}[1]
		\REQUIRE Values $v^{(\theta)}_1 > v^{(\theta)}_2 > ... > 0$, and discounts $\alpha^{(\theta)}_1 > \alpha^{(\theta)}_2 > ... > 0$ for each ad type $\theta$.
		\ENSURE Matching $M$ that maximizes $\sum_{(i, j)\in M} v_{ij}$.
		\STATE Initialize the dual solution so that \\
		$\quad \circ\  u_{i} \leftarrow 0$ for all ads $i$, \\
		$\quad \circ\ p_{j} \leftarrow \max v_{i',j'}$ for all slots $j$.

		\STATE Let $M \leftarrow \emptyset$ be the matching.
		\FOR{slot $j$ in descending order} \label{ln:for-slots}
		\STATE Let $B \gets \{j\} $ be an alternating BFS tree
		\STATE Let $P$ be an empty priority queue
		\STATE $P \leftarrow \textsc{UpdatePossibleNewEdges}(P, v, \alpha, M, j)$ \label{ln:init-P}
		\WHILE{$B$ does not contain an alternating path}    \label{ln:while-loop}
		\STATE $(i', j') \gets $ remove from $P$ next tight edge

		\STATE $\Delta \gets v_{i'j'} - u_{i'} - p_{j'}$ \label{ln:find-min} \COMMENT{note that $\Delta$ could be 0}
		\STATE Implicitly update the dual solution so that \\
		$\quad \circ\  u_{i''} \leftarrow u_{i''} + \Delta$ for all ads $i''\ \in B$, \\
		$\quad \circ\ p_{j''} \leftarrow p_{j''} - \Delta$ for all slots $j''\in B$.
		\IF{$i'$ is matched in $M$}
			\STATE $B \gets B \cup \{ (i', j'), (i', M(i')) \}$ \label{ln:expand-B}
			\STATE \textsc{UpdatePossibleNewEdges}$(P, v, \alpha, M(i'))$ \label{ln:update-P}
		\ELSE
			\STATE  $B \gets B \cup \{ (i', j') \}$  \COMMENT{now we have an augmenting path}

		\ENDIF
		\ENDWHILE
		\STATE  $A \gets $ alternating path in $B$
		\STATE $M \leftarrow \textsc{AddAlternatingPath}(M, A)$.
		\STATE explicitly update the dual solution $(u, p)$
		\ENDFOR
	\end{algorithmic}
	\caption{Hungarian algorithm for the Ad Types problem. \label{alg:ha}}

\end{algorithm}

Algorithm~\ref{alg:ha} shows how to compute the optimal allocation in an Ad Types instance. The algorithm is identical to how the Hungarian algorithm is commonly implemented; the only difference is how we maintain the set of possible new edges in Lines \ref{ln:init-P} and \ref{ln:update-P}. The algorithm initializes the dual solution $(u, p)$ to be feasible, and starts with an empty matching $M$. Algorithm considers slots in descending order in each iteration of the for loop in Line \ref{ln:for-slots}; we call each such iteration a \emph{phase}.

During each phase we iteratively update the dual variables until we find an augmenting path to increase the size of the matching $M$ by one. In each of these iterations within a phase we explore a tight edge leading to a matched edge and both edges are added to our alternating tree. Every time we add a new matched slot $j'$ to the alternating tree we explore the edges incident on $j'$ using the routine {\sc UpdatePossibleNewEdges}, which scan the edges incident on $j'$ and works out which edges are tight and when the remaining edges will become tight. All these new edges are stored in a priority queue for later retrieval.

\paragraph{High-level running time analysis.} Even though the algorithm is not well defined yet (we have not specified how to implement {\sc UpdatePossibleNewEdges}), still we can say something about the running time of the algorithm.

Each phase is implemented using a priority queue $P$ over some of the ads not in $B$. For each ad $i'$ in $P$ we keep track of the next edge $(i', j')$ that would become tight given the current structure of $B$. The priority of $i'$ captures \emph{when} this next edge becomes tight, the smaller the priority the sooner it becomes tight; similarly, if $i'$ already has a tight edge incident on itself then it should have the smallest priority in the queue.

In the normal implementation of the Hungarian Algorithm, the procedure \textsc{UpdatePossibleNewEdges}$(P, v, \alpha, j')$ iterates over all edges $(i', j')$ incident on $j'$. If $i' \in B$ we can ignore the edge as $i'$ has already been discovered and its slack $v_{i',j'} - \alpha_{j'} - u_{i'}$ will not change with future updates (since now both $i'$ and $j'$ belong to $B$). If $i' \notin B$ then we compute its current slack $v_{i',j'} - \alpha_{j'} - u_{i'}$ to work out when it will become tight and compare this against the time of the current next tight edge incident on $i'$, which we may need to update.

Without making any assumptions on the structure of the valuations, in the worst case in each iteration of the while loop in Line~\ref{ln:while-loop} we perform $O(nk + \log nk) = O(nk)$ work (assuming a Fibonacci heap implementation for $P$) since there are $kn$ ads in total and $kn$ edges incident on $j'$ (one per ad). In each iteration be grow $B$ by adding one new matched edge, so we have at most $j$ iterations of the while loop. Therefore, the overall running time is $O(\sum_{j=1}^n j nk) = O(n^3 k)$.

However, we can come up with a more efficient implementation of \textsc{UpdatePossibleNewEdges}$(P, v, \alpha, j')$ that exploits the special structure of our valuation function so that~$P$ holds at most $n+k$ ads and only $O(k)$ edges incident on $j'$ need to be scanned without sacrificing the overall correctness of the algorithm. With this improvement in performance, each iteration of the while loop in Line~\ref{ln:while-loop} takes at most $O(\log n + k)$ work. Again, since in each iteration be grow $B$ by adding one new matched edge, we have at most $j$ iteration of the while loop. Therefore, the overall running time is $O(\sum_{j=1}^n j (k + \log n)) = O(n^2 (k + \log n))$.

\begin{theorem}
\label{thm:runtime}
    Given an input with $k$ ad types and $n$ slots, Algorithm~\ref{alg:ha} can be implemented to run in time $O(n^2 (k + \log n))$.
\end{theorem}


Our goal for the the rest of this section is to provide an implementation \textsc{UpdatePossibleNewEdges}$(P, v, \alpha, j')$ where the size of $P$ is always at most $n+k$ and only $O(k)$ edges are considered in each invocation of the routine. Key to our analysis is the observation that tight edges cannot cross is the following sense: Given two ads $i<i'$ of the same type $\theta$, and two slots $j<j'$, then we cannot have the edge from ad $i$ to slot $j'$ be tight, and simultaneously have the edge from ad $i'$ to $j$ be tight.

\begin{lemma}[Non-crossing lemma]\label{lem:crossings}
	Given two ads $i<i'$ of the same type $\theta$, and two slots $j<j'$, if $v_i > v_{i'}$ and $\alpha^{(\theta)}_j > \alpha^{(\theta)}_{j'}$ then in any feasible dual solution we cannot have the edge from ad $i$ to slot $j'$ be tight, and simultaneously have the edge from ad $i'$ to $j$ be tight.
\end{lemma}

\begin{proof}
	We prove by contradiction. If the edges between $i$ and $j'$ and $i'$ and $j$ are both tight, then we must have dual variables $u_i, u_{i'}, p_j, p_{j'}$ such that
	\begin{align*}
	\alpha^{(\theta)}_{j'} \cdot v_i &= u_i + p_{j'}\\
	\alpha^{(\theta)}_j \cdot v_{i'} &= u_{i'} + p_{j}
	\end{align*}
	At the same time, due to the slackness constraints, we must have that
	\begin{align*}
	\alpha^{(\theta)}_j \cdot v_i &\le u_i + p_j\\
	\alpha^{(\theta)}_{j'} \cdot v_{i'} &\le u_{i'} + p_{j'}.
	\end{align*}
	We can combine these and obtain
	\begin{align*}
	\alpha^{(\theta)}_{j'} \cdot v_i + \alpha^{(\theta)}_j \cdot v_{i'} &= u_i + p_{j'} + u_{i'} + p_{j} \\
	&\ge \alpha^{(\theta)}_j \cdot v_i + \alpha^{(\theta)}_{j'} \cdot v_{i'}.
	\end{align*}
	Which is false due to the standard exchange argument. We give the argument for completeness: Rearranging we have $(\alpha^{(\theta)}_j - \alpha^{(\theta)}_{j'})\cdot(v_i - v_{i'}) \le 0$; however, due to strict monotonicity $\alpha^{(\theta)}_{j} > \alpha^{(\theta)}_{j'}$ and $v_i > v_{i'}$, so have reached a contradiction.
\end{proof}

\subsubsection{\textsc{UpdatePossibleNewEdges}}

The goal of \textsc{UpdatePossibleNewEdges}$(P, v, \alpha, j')$ is to iterate over the edges incident on $j'$ that are tight or that can potentially become tight later in the execution of the current phase. For each such edge $(i', j')$ we compare its slack with the priority associated with $i'$ and update the entry for $i'$ in $P$ accordingly if needed.

The exact definition of the edges inspected is given by Algorithm~\ref{alg:new-edges}. Instead of describing how this works, let us make some observations about the set of edges that can potentially become tight, and then we shall see that the Algorithm indeed considers all these edges.

For each ad type $\theta$ we first consider the edges of the form $(a^{(\theta)}_i, j')$ where $a^{(\theta)}_i$ is matched and $M(a^{(\theta)}_i) < j'$. We claim that we only need to consider the largest such $i$. Recall that all the edges in $M$ are tight and remain tight throughout the execution of the phase; in particular, $(a^{(\theta)}_i, M(a^{(\theta)}_i))$ is tight and remains tight. Thus, any edge $(a^{(\theta)}_{i'}, j')$ with $i' < i)$ is not tight and will never become tight due the Non-crossing Lemma~\ref{lem:crossings} and the fact that $i' < i$ and $M(a^{(\theta)}_i) < j'$.

Now consider the edges of the form $(a^{(\theta)}_i, j')$ where $a^{(\theta)}_i$ is matched and $M(a^{(\theta)}_i) > j'$. We claim that we only need to consider the smallest such $i$. Recall that all the edges in $M$ are tight and remain tight throughout the execution of the phase; in particular, $(a^{(\theta)}_i, M(a^{(\theta)}_i))$ is tight and remains tight. Thus, any edge $(a^{(\theta)}_{i'}, j')$ with $i' > i$ is not tight and will never become tight due the Non-crossing Lemma~\ref{lem:crossings} and the fact that $i' > i$ and $M(a^{(\theta)}_i) > j'$.

Finally, we need to consider edges of the form $(a^{(\theta)}_i, j')$ where $a^{(\theta)}_i$ is unmatched. We claim that we only need to consider the smallest such $i$ available\footnote{It is worth noting that even this case can be ignored if there exists a matched $a^{(\theta)}_i$ such that $M(a^{(\theta)}_i) > j'$; however, for ease of presentation we add the slot to $X$ even if such $a^{(\theta)}_i$ exists.}. Indeed, for any other $i' > i$ note that $v^{(\theta)}_i > v^{(\theta)}_{i'}$ and since the $u$ variable of both ads is 0 (only slots that are part of an alternating tree get their dual variables increased and those are always matched) the slack of $(a^{(\theta)}_i, j')$ will always be smaller than the slack of $(a^{(\theta)}_{i'}, j')$ since
$v^{(\theta)}_i \alpha^{(\theta)}_{j'}  - p_{j'} < v^{(\theta)}_{i'} \alpha^{(\theta)}_{j'} - p_{j'}$. Furthermore, notice that if the edge $(a^{(\theta)}_{i}, j')$ become tight, then we immediately have an augmenting path in $B$, which concludes the phase.

It is easy to see that these three cases are precisely those covered by Algorithm~\ref{alg:new-edges}.

\begin{algorithm}
	\begin{algorithmic}[1]
		\REQUIRE $P, v, \alpha, M, j'$
		\STATE $X \gets \emptyset$
		\FOR{ad type $\theta$}
			\STATE let $a^{(\theta)}_i$ be the unmatched ad of type $\theta$ with smallest $i$
			\STATE add $a^{(\theta)}_i$ to $X$

			\IF{exists matched ad $a^{(\theta)}_i$ such that $M(a^{(\theta)}_i) < j'$}
			  \STATE let $a^{(\theta)}_i$ be such an ad with largest $i$
				\STATE add $a^{(\theta)}_i$ to $X$
			\ENDIF
			\IF{exists matched ad $a^{(\theta)}_i$ such that $M(a^{(\theta)}_i) > j'$}
				\STATE let $a^{(\theta)}_i$ be such an ad with smallest $i$
				\STATE add $a^{(\theta)}_i$ to $X$
			\ENDIF
		\ENDFOR
		\FOR{$i' \in X$} \label{ln:for-X}
			\IF{$i' \notin B$ and either $i' \notin P$ or $i'$' current slack in $P$ is $ > (v_{i', j'} - u_{i'} - \alpha_{j'})$ }
				\STATE update the priority of $i'$ using $(i', j')$ or set if $i' \notin P$ \label{ln:update-priority}
			\ENDIF
		\ENDFOR
	\end{algorithmic}
	\caption{\textsc{UpdatePossibleNewEdges}}
	\label{alg:new-edges}
\end{algorithm}

\begin{lemma}
		There can be at most $n+k$ ads in $P$ at any given point in time.
\end{lemma}

\begin{proof}
		Notice that the the only edges $(i', j')$ that we consider in Line~\ref{ln:update-priority} are either to a matched node in $M$ or to the heighest unmatched ad of each type. There are exactly $j < n$ matched ads in phase $j$ and there are $k$ ad types, so the lemma follows.
\end{proof}

\begin{lemma}
	\textsc{UpdatePossibleNewEdges} considers only $O(k)$ edges when updating $P$ and these are the only edges we need to look at.
\end{lemma}

\begin{proof}
		The reason why we can focus just on the edges in Line~\ref{ln:for-X} has already been explained in the description of the algorithm \textsc{UpdatePossibleNewEdges}, so we only need to argue about their number. Clearly, each ad type generated at most three ads into $X$, thus, the number of edges considered when updating $P$ is at most $3k$.
\end{proof}

\subsection{Large Number of Ad Types}
\label{ss:kn}
When $k=n$, Algorithm~\ref{alg:ha} is no faster than running the standard Hungarian algorithm. The following lemma shows that this is to be expected as any instance of the assignment problem can be reduced to an instance where all ads agree on the order of the slots.

\begin{lemma}
	With $k=n$ ad types, the problem is no easier to solve than the assignment problem, even with monotone discount curves.
\end{lemma}
\begin{proof}
	We'll reduce to the assignment problem. Given an instance of the assignment problem (with positive edge values), let $v^* = 1 + \max_{(i,j) \in E} v_{ij}$, indicate one of the sides of the bipartite graph as the slots, and choose an order $\sigma$. For all incident edges to the lowest slot, add $0$ to the edge. For the second to last edge, add $v^*$ to the edge, for the third to last edge add $2v^*$, and so on. 	The value of any perfect matching has increased by $\frac{m(m-1)}2v^*$, which is a constant. So the solution to the optimal perfect matching is the same as the original problem. Note that now for each ad, the slots are monotonically decreasing according to order $\sigma$, so for each ad there's a decomposition into a base value, and a discount for each slot which is monotonically decreasing.
\end{proof}

\subsection{Computing VCG Prices}
\label{ss:vcg}
Algorithm~\ref{alg:ha} gives us a fast way to compute the optimal allocation for the Ad Types problem. The VCG mechanism \cite{V61, C71, G73} can be use to compute corresponding incentive compatible prices, by running the allocation algorithm $n+1$ times, for a total running time of $O(n^3(k + \log n))$. In this section we show that in fact we can compute VCG prices in asymptotically identical time to solving the allocation problem once: $O(n^2(k + \log n))$. It is a folk theorem that an instance of the \emph{assignment problem} can be priced in asymptotically the same time as solving the assignment problem (namely $O(n^3)$ for balanced bipartite graphs); this section shows that our improved running time carries over as well.

We begin with the well-known fact (e.g. \cite[Chapter 15]{EK10}) that VCG prices and buyer utilities are feasible (and optimal) dual variables for the assignment problem. Additionally, the VCG prices correspond to the \emph{point-wise lowest} possible dual variables for the slots over all feasible solutions to the dual. Thus, we can find VCG prices by finding the lowest dual variables.

\begin{theorem}
\label{thm:vcg}
Given an instance $(U, V, E)$ of the assignment problem, along with an optimal matching $M$, and shadow prices $\bf p$. Algorithm~\ref{alg:vcg} computes VCG prices in $O(|E|+|V|\cdot \log|E|)$ time.
\end{theorem}
\begin{proof}
    Algorithm~\ref{alg:vcg} is similar to a single ``phase'' of Algorithm~\ref{alg:ha}. To find $\Delta$ we use a Fibonnacci heap on the edges where $(i\in R, r\not\in R)$ or $(i\not\in R, r\in R)$. During each while loop, we mark at least one slot as immutable, so the total number running time to find $\Delta$ over the entire algorithm is $O(|V|\log |E|)$.
    
    The alternating BFS trees can only touch each edge ones over the course of all iterations, hence the total time this takes is bounded by $O(|E|)$, finally each node is marked as immutable at most once, hence the total time over all iterations of line 6 is $O(|V|)$. This yields total running time of $O(|E| + |V|\log |E|)$ as claimed. Note that no price could be lower, since all nodes either have price $0$ or lowering their price would violate a slackness constraint. Thus, the dual variables for the slot correspond to VCG prices.
\end{proof}

\begin{algorithm}
	\begin{algorithmic}[1]
		\REQUIRE $U, V, E, M, \bf p$
		\ENSURE ${\bf p}_{vcg}$
        \STATE Let $R$ be the set of nodes in the matching
         \WHILE{$R \neq \emptyset$}
		\STATE Let $\Delta = \min(\min_{i\in R}p_i, \min_{(i,j)\in E, (i\in R, r\not\in R) \text{ or } (i\not\in R, r\in R)} v_{ij} - u_i - p_j)$
		\STATE Implicitly update the dual solution so that \\
		$\quad \circ\  u_{i''} \leftarrow u_{i''} + \Delta$ for all ads $i''\ \in R$, \\
		$\quad \circ\ p_{j''} \leftarrow p_{j''} - \Delta$ for all slots $j''\in R$.
		\STATE Do alternating BFS $B$ on tight edges from all lowest prices, i.e. the slots that have $p_j = 0$ or have a new tight edge to an ad outside $R$ or the ads that have a new tight edge to an ad outside $R$.
		\STATE Mark all visited nodes in $B$ as ``immutable'', remove them from $R$.
	    \ENDWHILE
	    \STATE Explicitly update the dual variables.
	    \RETURN {\bf p}
	\end{algorithmic}
	\caption{Compute VCG prices for a given solution to an instance of the assignment problem.}
	\label{alg:vcg}
\end{algorithm}

\begin{corollary}
    Given an instance of the Ad Types problem, we can compute the optimal allocation and VCG prices in $O(n^2(k + \log n))$ time.
\end{corollary}
\begin{proof}
    By theorem~\ref{thm:runtime} we can compute an optimal allocation in $O(n^2(k + \log n))$ time. By Theorem~\ref{thm:vcg} and since $|E| = kn^2$, from the optimal allocation we can compute the corresponding VCG prices in $O(kn^2 + n\log kn)$ hence the total time is $O(n^2(k + \log n))$ as claimed.
\end{proof}

\subsection{Applying Reserve Prices}
\label{ss:reserve-prices}
Setting a reserve price, or a minimum bid, is a common way to control revenue \cite{M81}. But how can we apply a reserve price in the ad types problem? In our setting, where we have discrete allocation amounts (the discounts $\alpha_j^\theta$ for an ad of type $\theta$), Myerson's payment identity boils down to determining the changepoints in the mechanism, given the other bids $b_{-i}$. We could find these changepoints by updating the bid of an advertiser while maintaining the invariants in the Hungarian algorithm, but we present a simpler approach here that may be of independent interest to the mechanism design community. The idea is similar to the externality computation in VCG, but rather than removing the advertiser in the alternative allocation, we compare against an alternative where the advertiser bids their reserve price.

Since this approach works in single-parameter environments beyond the Ad Types Problem, we use the more general notation introduced in Section~\ref{ss:auction-theory}. Recall that we have $n$ strategic agents, outcome space $\Omega$, function $x_i : \Omega \rightarrow \RR$ from outcomes to a quantity of ``stuff'', and valuation function $v_i(\omega) = x_i(\omega)\cdot v_i$ where we overload notation and let $v_i\in \RR$. Additionally, we have agent-specific reserve prices $r_i$ below which we won't serve an agent. Finally, we have black-box access to an algorithm $\XX : \RR^{n} \rightarrow \Omega$ that given $n$ bids $\bf b$, yields an outcome $\omega$ that maximizes $\sum_{i=0}^n v_i(\omega)$, the social welfare of strategic agents.

The goal in this setting is to come up with an incentive-compatible pricing scheme to match the social welfare-maximizing (subject to reserve prices) allocation. Due to Myerson, we know that if we can compute the change points in the allocation function, we can recover the (unique up to translation) pricing rule that is incentive compatible. However, generally computing change points may be costly. The goal then is to compute incentive compatible prices without change point computation.

\paragraph{Incentive Compatible Pricing without Changepoint Computation.}
The algorithm is given in Algorithm~\ref{alg:pricing}, where we overload $\XX$ to take the bids of the set passed to it, and we use the notation $\XX(b_i\rightarrow r_i, B'_{-i})$ to indicate applying $\XX$ to the set $B'_{-i}$ plus bidder $i$ with their bid replaced by $r_i$. In the traditional VCG algorithm, the buyer's externality is computed by removing them from the auction entirely and charging the change in utility of the remaining buyers. The difference here is that we're not removing the buyer, but instead we replace their bid by the reserve price. Additionally we charge them $r_i$ per unit of stuff they receive if they bid their reserve price. Intuitively this payment scheme corresponds to Myerson payments, where we only explicitly compute the change point at the reserve price, and we implicitly compute reserve prices by looking at the difference in outcomes between bidding $b_i$ and $r_i$.

\begin{algorithm}
	\begin{algorithmic}[1]
		\REQUIRE Set of buyers $B$, with bids ${\bf b} \in \RR^n$, reserves ${\bf r} \in \RR^n$.
		\ENSURE Outcome $\omega$, prices $p_i \in \RR^n$.
		\STATE $B' \leftarrow \{i \in B: b_i \ge r_i \}$
		\STATE For buyer $i\not\in B'$, $p_i \leftarrow 0$
		\STATE $\omega \leftarrow \XX(B')$
		\FOR{buyer $i\in B'$}
		\STATE $SW_{-i}(B') \leftarrow \sum_{j \in B' \backslash \{i\}} v_j(\omega)$
		\STATE $\omega' \leftarrow \XX(b_i\rightarrow r_i, B'_{-i})$
		\STATE $SW_{-i}(b_i\rightarrow r_i, B'_{-i}) \leftarrow \sum_{j \in B' \backslash \{i\}} v_j(\omega')$
		\STATE $p_i \leftarrow SW_{-i}(b_i \rightarrow r_i, B'_{-i}) - SW_{-i}(B') + x_i(\omega')\cdot r_i$
		\ENDFOR
		\RETURN $(\omega, \bf p)$
	\end{algorithmic}
	\caption{Incentive-compatible pricing without change point computation.}
	\label{alg:pricing}
\end{algorithm}

Myerson \cite{M81} showed that the payments that yield an incentive compatible auction are uniquely determined given an allocation rule, up to an additive constant. Of special interest are payments where losers pay nothing. Therefore, it suffices to show that the payments computed by Algorithm~\ref{alg:pricing} are incentive compatible and charge losers nothing to prove that they are equivalent to applying Myerson's payment identity. In the following let $SW(B) =  \sum_{j \in B}v_j\cdot x_j(\omega)$ for $\omega = \argmax_{\omega \in \Omega} \sum_{j \in B}b_j\cdot x_j(\omega)$---i.e. the social welfare of the optimal allocation assuming bids equal values---with $SW(B) = SW_i(B) + SW_{-i}(B)$ such that $SW_{i}$ is only the contribution of $i$ to $SW(B)$ and $SW_{-i}(B)$ is the contribution of the remaining bidders $\{j \in B: j\neq i\}$. Similarly, let $SW(b_i\rightarrow r_i, B_{-i}) =  \sum_{j \in B}v_j\cdot x_j(\omega')$ for $\omega' = \argmax_{\omega \in \Omega} r_i\cdot x_i(\omega) + \sum_{j\in B, j\neq i}b_j\cdot x_j(\omega)$ i.e. the social welfare of the optimal allocation assuming bids equal values and $b_i$ replaced by $r_i$, with $SW_i(b_i\rightarrow r_i, B_{-i})$ and $SW_{-i}(b_i \rightarrow r_i, B_{-i})$ defined analogously.

\begin{theorem}
Algorithm~\ref{alg:pricing} yields the unique Incentive Compatible pricing scheme for $\XX$ with eager reserves that charges losers $0$.
\end{theorem}
\begin{proof}
    \begin{align*}
        u_i(b_{-i}, b_i) &= \begin{cases}
            x_i({\bf b}) v_i - \left( SW_{-i}(b_i\rightarrow r_i, B'_{-i}) - SW_{-i}(B') + x_i(\omega')\cdot r_i\right) & \text{if }b_i \geq r_i \\
            0 & \text{otherwise}\\
            \end{cases}\\
            &= \begin{cases}
            SW_i(B') - \left( SW_{-i}(b_i \rightarrow r_i, B'_{-i}) - SW_{-i}(B') + SW_i(b_i \rightarrow r_i, B'_{-i})\right) & \text{if }b_i \geq r_i \\
            0 & \text{otherwise}\\
            \end{cases}\\
            &= \begin{cases}
            SW(B') - SW(b_i\rightarrow r_i, B'_{-i}) & \text{if }b_i \geq r_i \\
            0 & \text{otherwise}\\
            \end{cases}
    \end{align*}

    We treat the cases for $v_i < r_i$ and $v_i \ge r_i$ separately.

    If $v_i<r_i$, then $u_i(b_i, b_{-i}) = 0$ for all $b_i<r_i$, so there's no beneficial deviation from $b_i = v_i$ to some other $b_i < r_i$. Moreover, bidding $b_i> r_i$ can only yield negative utility: first observe that if $b_i = r_i$, it must hold that $SW(B') \le SW(b_i\rightarrow r_i, B'_{-i})$ by definition since the latter is the optimal social welfare achievable (if $b_i = r_i$), hence the former can only be lower. This implies that the utility must be non-positive for any $v_i \le r_i$ since the value per-quantity good is lower than if $v_i = r_i$ but the price stays the same. So for $v_i<r_i$ there is no beneficial deviation.

    If $b_i\ge r_i$, notice that the term $SW(b_i\rightarrow r_i, B'_{-i})$ is independent of the bid of bidder $i$. So the only relevant term is $SW(B')$. Since $\XX$ compute the maximum-weight bipartite matching given the bids, bidding $v_i$ yields the maximum utility; hence no other bid $b_i \ge r_i$ can improve utility. Moreover, since the utility is guaranteed to be non-negative, bidding $b_i<r_i$ (which yields $u_i = 0$) cannot increase utility.

    So in all cases, a bidder maximizes their utility by bidding $v_i$. Moreover, in both cases, if $x_i(\omega) = 0$, the bidder pays nothing, so Algorithm~\ref{alg:pricing} yields the unique incentive compatible prices where losers pay 0.
\end{proof}

Thus, if we use Algorithm~\ref{alg:ha} to compute the max-weight bipartite matching after applying reserve prices, Algorithm~\ref{alg:pricing} runs in $O(n^3(k + \log n))$.

\begin{corollary}
    \label{cor:myerson}
    Algorithm~\ref{alg:pricing} runs in $O(n^3(k + \log n))$ for the Ad Types Problem without gap rules.
\end{corollary}

%% file: gap-rules.tex

\section{The Ad Types Problem with Gap Rules}

In this section we switch our attention to the full version of the Ad Types problem where we do have gap rules. As we shall see, this problem is much harder if we do not place any restriction on the instances. On the positive side, we show that the problem is fixed parameter tractable on the number of ad types.

\subsection{Hardness}

In this subsection we prove that even approximating an optimal allocation for the Ad Types problem with Gap Rules is hard, even when there are no discount curves.

\begin{theorem}
\label{thm:hard}
  The Ad Types problem with Gap Rules hard to approximate better than $k^{1-\epsilon}$ for any $\epsilon > 0$, unless P = NP, even when the discount curves of all the ad types are identically equal to $1$.
\end{theorem}

\begin{proof}
The proof reduces the Maximum Independent Set (MIS) problem to the Ad Types problem.
In the MIS problem we have a graph $(V,E)$ and we want to find the largest subset of vertices such that no two vertices are connected with an edge.

Given an instance of the MIS problem, we can build an instance of the Ad Types problem as follows:
\begin{itemize}

    \item Let the number of slots and the number ad types equal the number of vertices in the MIS problem, i.e., $k = n = |V|$.

    \item Each ad type has exactly one ad in it, i.e., for each $\theta_l$ for $l \in \{1, \ldots, k\}$, $|A^{\theta_l}| = 1$.

    \item Each ad has value-per-conversion equal to $1$, i.e. $\forall l \in \{1, \ldots, k\}$, $v^{(\theta_l)} = 1$. Notice that we dropped the subscript over $v$, since there is only one ad per ad type.

    \item The ads have no discount curves over the slots, i.e., $\alpha_1^{\theta_l} = \ldots = \alpha_k^{\theta_l} = 1$ $\forall l \in \{1, \ldots, k\}$.

    \item Let $G$ be the $k \times k$ matrix describing the gap rule. The gap between two ad types $i$ and $j$ is $k$ if $(i,j) \in E$ and $0$ otherwise, i.e., $G_{ij} = k$ if $(i,j) \in E$ and $G_{ij} = 0$ otherwise.

\end{itemize}

Given the construction of the gap rules, two ads can belong to a solution of the Ad Types problem only if their corresponding vertices are not connected in $G$.
Since all the ads have value $1$ for all the slots, a solution of the Ad Types problem represents how many ads have been allocated. Moreover, each allocated ad contributes 1 to the wellfare objective of the Ad Types problem so maximizing the value of the allocation for the Ad Types problem is equivalent to maximizing the size of the correspondign independent set. Thus, the Ad Types problems inherets the $k^{1-\epsilon}$ hardness of approximation of MIS~\cite{H96,Zuckerman07} since $|V|=k$.

\end{proof}

\subsection{Exact Algorithm}

Since the previous section shown that the Ad Types Problem with Gap Rules cannot be solved in polynomial time, unless $P=NP$, here we provide a dynamic-programming based algorithm that solves the problem exactly in exponential time.

The dynamic program has $2k$ dimensions. For each ad type $l$ the dynamic program stores (i) the number of ads of type $\theta_l$ that have been allocated, and (ii) the largest slot index that an ad of type $\theta_l$ has been allocated.

As usual, we assume that within each ad type, the ads monotonically decrease in valuation, i.e., for each $\theta_l,$ $v_1^{\theta_l} \geq v_2^{\theta_l} \geq \ldots \geq v_{n}^{\theta_l}$.

In order to describe our dynamic program, we need some additional notation:

\begin{itemize}
    \item let $n^{\theta_l}$ be the number of ads allocated for the ad type $\theta_l$, we denote these values with a vector $\bm{n^\theta} = \langle n^{\theta_1}, \ldots, n^{\theta_k}\rangle$
    \item let $s^{\theta_l}$ be the last slot where an ad of type $\theta_l$ has been allocated, we denote these values with a vector $\bm{s^\theta} = \langle s^{\theta_1}, \ldots, s^{\theta_k}\rangle$
    \item for a given vector $\bm{a}$, let $\bm{a_{-l,x}}$ to denote a new vector where the $l$th coordinate is replaced by $x$, i.e., if $\bm{a} = \langle a_1, \ldots, a_k\rangle$ then $\bm{a_{-l,x}} = \langle a_1, \ldots, a_{l-1}, x, a_{l+1}, \ldots, a_k\rangle$
\end{itemize}

Let $T[\bm{n^\theta}, \bm{s^\theta}]$ be our dynamic program states, which we define to be maximum value among allocations that assigns $n^{\theta_l}$ ads of type $\theta_l$ with the last ad being assigned to slot $s^{\theta_l}$. We allow $s^\theta_l$ to be empty (which we can denote with the special symbol $\perp$) to handle the case where $n^\theta_l = 0$.

To derive the recurrence for our DP formulation, let us focus on a type $\theta_l$ with largest $s^\theta_l$ value. Given that the last ad of type $\theta_l$ is in $s^\theta_l$, let us denote with $P(\theta_l, \bm{n^\theta}, \bm{s^\theta})$ the set of possible slots where we could have the second to last ad of type $\theta_l$. If $n^\theta_l = 1$ then $P(\theta_l, \bm{n^\theta}, \bm{s^\theta}) = \{ \perp \}$, otherwise, it contains every slot that obeys all the gap rules between this new ad and the ads locations specified by $\bm{s^\theta}$.

\begin{align}
    T[\bm{n^\theta}, \bm{s^\theta}] = \max_{j \in P(\theta_l, \bm{n^\theta}, \bm{s^\theta})}  \Big\{ T[\bm{n^\theta_{-l,n^\theta_l - 1}}, \bm{s^\theta}_{-l, j}] + v_{n^{\theta_l}}^{\theta_l} \cdot \alpha_{s^{\theta_l}}^{\theta_l} \Big\}\text{,
    where } \theta_l = \max_{\theta_l} s^{\theta_l} \label{eq:dp}
\end{align}

The base case of the recurrence is when $\bm{n^\theta}$ is the all 0 vector and $\bm{s^\theta}$ is the all empty-position vector. If we ever reach a state where $P(\theta_l, \bm{n^\theta}, \bm{s^\theta}) = \emptyset$ then the state is infeasible so we set its cost to be $-\infty$. The optimal solution will be at a state with highest value.

\begin{theorem}
    \label{thm:dp}
    The Ad Types Problem can solved optimally in $O(k\cdot n^{2k + 1})$ time.
\end{theorem}

\begin{proof}
  The correctness of the dynamic programming formulation follows from the fact that for any of the states on the right hand side of \eqref{eq:dp} any allocation that is consistent to the state can be augmented by allocating the $n^{\theta_l}$th ad of type $\theta_l$ to slot $s^{\theta_l}$ due to the fact that $\theta_l$ has the largest $s^{\theta_l}$ value.

  We only consider states where $s^{\theta_l} \leq n$ and $n^{\theta_l} \leq n$ for all types $\theta_l$, so the DP table has $n^{2k}$ entries. Each entry takes $O(kn)$ time to compute (there are $n$ possible slots in $P(\cdot)$ and each one takes $O(k)$ to check whether it obeys the gaps rules), so the total running time is $O(k\cdot n^{2k+1})$.
\end{proof}